\documentclass[conference]{IEEEtran}

\usepackage[utf8]{inputenc}
\usepackage[draft,author=,mode=multiuser]{fixme}
\usepackage{algpseudocode}\usepackage[dvipsnames,table,xcdraw]{xcolor}

\usepackage[english]{babel}

\usepackage{mathtools}
\DeclarePairedDelimiter\ceil{\lceil}{\rceil}
\DeclarePairedDelimiter\floor{\lfloor}{\rfloor}

\usepackage{amsthm}
\newtheorem{theorem}{Theorem}
\newtheorem{property}{Property}

\usepackage{tabto}

\usepackage{amsmath}
\usepackage{graphicx}
\usepackage[colorlinks=true, allcolors=blue]{hyperref}
\usepackage{diagbox}

\title{Blockchain Scalability and Security: Communications Among Fast-Changing Committees Made Simple}

\begin{document}
\author{
\IEEEauthorblockN{Andrea Mariani}
\IEEEauthorblockA{\textit{Engineering Dept.} \\
\textit{Roma Tre University}\\
Roma, Italy \\
\small and.mariani3@stud.uniroma3.it
}
\and
\IEEEauthorblockN{Gianluca Mariani}
\IEEEauthorblockA{\textit{Engineering Dept.} \\
\textit{Roma Tre University}\\
Roma, Italy \\
\small gia.mariani1@stud.uniroma3.it
}
\and
\IEEEauthorblockN{Diego Pennino}
\IEEEauthorblockA{\textit{Engineering Dept.} \\
\textit{Roma Tre University}\\
Roma, Italy \\
\small pennino@ing.uniroma3.it\\
0000-0001-5339-4531}
\and
\IEEEauthorblockN{Maurizio Pizzonia}
\IEEEauthorblockA{\textit{Engineering Dept.} \\
\textit{Roma Tre University}\\
Roma, Italy \\
\small pizzonia@ing.uniroma3.it\\
0000-0001-8758-3437}
}
\maketitle

\begin{abstract}
For permissionless blockchains, scalability is paramount. While current
technologies still fail to address this problem fully, many research works
propose sharding or other techniques that extensively adopt parallel processing
of transactions. In these approaches, a potentially large number of committees
of nodes independently perform consensus and process new transactions. Hence, in
addition to regular intra-committee communication, (1) new transactions have to
be delivered to the right committee, (2) committees need to communicate to
process inter-shard transactions or (3) to exchange intermediate results. To
contrast slowly adaptive adversaries, committees should be frequently changed.
However, efficient communication to frequently-changing committees is hard.

We propose a simple approach that allows us to implicitly select committee
members and effectively deliver messages to all members of a specific committee, even
when committees are changed frequently. The aim of our design is to provide a
committee selection procedure and a committee-targeted communication primitive
to be applied in most of the scalable blockchain architectures that are
currently proposed in literature. We provide a theoretical proof of the security
of our approach and first experimental results that shows that our approach might be feasible in practice.

\end{abstract}

\vspace{0.2cm}

\section{Introduction}

Permissionless blockchains allow any subject to join the network, either for
helping with transaction processing or for producing new workload. Clearly,
scalability is a fundamental aspect of them.

Many research contributions focus on the scalability of consensus protocols
(e.g., \cite{chen2019algorand,kiayias2017ouroboros,8123011, 8972381, oyinloye2021blockchain}). However,
scalability is not only related to consensus. Supposing each node in the network has the same limited amounts of
CPU power, storage and network bandwidth, the scalability of current blockchains 
technologies is severely limited by the fact that (accepted or candidate) transaction are broadcasted and that each node involved in consensus have to process each transaction and to store it (or its effect).

To overcome this fundamental problem the \emph{sharding} technique has been
proposed. The nodes of the sharded network are partitioned into \emph{shards}.
Each shard manages its own blockchain which stores a subset of the transactions/state of the whole system. In each shard, transactions are accepted by a
\emph{committee} of nodes, which in general may be a subset of the nodes of the
shard and may change over time. Transactions may involve addresses belonging to more shards hence inter-shard transactions should be handled by some sort of inter-committee communication. See also~\cite{zamani2018rapidchain, kokoris2018omniledger, dang2019towards, wang2019sok, 8954616, hashim2023sharding}.

Another approach is proposed in~\cite{DBLP:journals/corr/abs-2005-06665}, where
the blockchain is only one but several randomly selected committees cooperates
to obtain the set of accepted transaction of each block and to obtain its header
which is the only data sent in broadcast. In that work, committees cooperation is
inspired to the pipeline processing technique similarly to what occurs in modern
CPUs, where each committee executes a stage of the pipeline. Further, the first
stages of the pipeline also share the load among several committees.

In both approaches, if nodes and workload increase proportionally, scalability
is easily achieved by increasing the number of shards or committees that are active at the
same time.

For security reason, committees involved in consensus should change frequently.
However, it is not clear how inter-committee communication is achieved in this
case. In fact, we cannot use broadcast, for scalability reasons, and multicast
techniques need to rely on some configured routing, which is hard to update
efficiently when committees change.

In this paper, we show an approach for achieving both inter-committee
communication and high frequency of committee changes. Our solution has a number
of properties that makes it suitable for its adoption within the architectures
that have scalability as their primary concern, like, for
example,~\cite{DBLP:journals/corr/abs-2005-06665, zamani2018rapidchain}. Our
approach is based on kademlia but kademlia routing table does not have to be updated
when committee is changed. To send a message to all members of a committee, the
sender needs to know only a limited amount of pseudo random numbers called
\textit{centers} that are used to select the committee members. Committees
members are selected close to the centers and it is easy to verify if a certain
node belongs to a committee. We prove that security of our approach quickly
increases when the number of centers increases. Further, we experimentally show
that the probability of a successful attack quickly goes to zero even with small number of centers.
This is a desirable property since the efficiency of the communication protocol is affected by
the number of centers.

The structure of the paper is the following. 
In Section~\ref{sec:SoA},  we quickly review some related literature.
In Section~\ref{sec:requirements}, we state the problem and the requirements with basic definitions and assumptions. 
In Section~\ref{sec:solution}, we present our solution detailing procedures for committee selection and committee-targeted communication.
In Section~\ref{sec:robustness-analysis}, we show that committees contains enough members with high probability. 
In Section~\ref{sec:security-analisys}, we show a formal security analysis. 
In Section~\ref{sec:experiments}, we show our first experimental results.
In Section~\ref{sec:conclusions}, we draw the conclusions and we propose several research directions.

\section{State of the Art}\label{sec:SoA}

Recent blockchains that aim at low latencies are committee-based, in the sense that a
selection of nodes runs a consensus algorithm with deterministic finality (e.g.,
\cite{EOSIO14:online, chen2019algorand}). To comply with
high standards of security and decentralization the committee should change very
frequently: for example Algorand~\cite{chen2019algorand} changes its
committee after an interval whose duration is in the order of seconds.

When scalability is the main concern, more than one committee is active at the
same
time~\cite{luu2016secure,kokoris2018omniledger,zamani2018rapidchain,DBLP:journals/corr/abs-2005-06665} so that workload can be shared in some way over several committees. In this setting, a routing system to reach a specific committee is usually
needed. Since routing information takes time to propagate, changing the
committee frequently is hard. For example,
RapidChain~\cite{zamani2018rapidchain}, one of the most famous sharding
proposals, allow changes of small parts of each committee only every 24 hours. It is likely that this
interval of time can be greatly reduced without hurting functionality, however,
it should be well above average propagation time for routing information.

\textit{Sharding} is the most promising approach for creating scalable distributed ledgers.
Transactions and blockchain state are the partitioned into
multiple shards managed in parallel by different subsets of validators. However, handling 
\textit{inter-shard transactions} is problematic.
Typically, this introduces communication inefficiencies between shards and the
need for techniques similar to a \textit{two-phase commit} to ensure that
transaction executions are atomic. For example, in OmniLedger
\cite{kokoris2018omniledger} inter-shard transactions are handled by clients
acting as gateways to support message passing between shards (this means there
is no direct communication between shards). However, this overloads the clients
too much, violating the principles of light clients. Additionally, the clients
need to broadcast the transactions to all network.

RapidChain~\cite{zamani2018rapidchain} introduces an inter-committee routing
protocol. Each committee is in charge of a shard. Each committee and
transactions have an ID. Each committee $C$ deals with all transactions that
have the ID of $C$  as  prefix of their ID. The node that receive a transaction
sends it to the responsible committee through the inter-committee routing
protocol. The routing algorithm is inspired by Kademlia
\cite{maymounkov2002kademlia}. In particular, each node stores information about
all members of the committee it belongs to and information about only a small
sample of nodes ($O(\log \log n)$) from the other $O(\log n)$ committees whose
ID is XOR close to the ID of the committee it belongs to. Then, based on the
prefix ID of the transaction, a node will send the transaction to the closest
committee (i.e, to the  few nodes which he knows in that committee) to the
responsible committee that he knows. Hence, routing is done at committee level
and it gets close the right committee in $O(\log n)$ steps. RapidChain varies
their committees infrequently and only in part, since committee change is
inefficient with respect to network load.

\section{Requirements} \label{sec:requirements}

In this section, we state the problem and the requirements that our solution should comply with.

We assume to have a universe $U$ of
\emph{potential committee members} whose cardinality is $N=|U|$. \emph{Committees} are subsets of $U$ denoted
$C_i$, with $1<i<K$, with $K$ the number of committees. Their cardinality
$|C_i|$ is independent of $|U|$ and, ideally, it is constant, In practice, we
can accept it to be a random variable with certain properties (see the
\emph{robustness} requirement). To scale, $K$ may grow
linearly with $|U|$. We call \emph{round} a period of time in which each
committee performs its task to reach a consensus on some result. The specific
task and result is not relevant for this paper: it might be the proposal of the
next block or a stage of a pipeline as
in~\cite{DBLP:journals/corr/abs-2005-06665}. At the end of the round all
committees are changed. Rounds are numbered, hence for each round $r$ each
committee is denoted $C_{i,r}$.

We look for a solution comprising two elements:
\begin{itemize}

\item a committee selection algorithm with certain security properties (stated
in Section~\ref{sec:requirements}), 

\item a routing mechanism that allows to reach all members of each committee
$C_{i,r}$ without the need to update any routing information.

\end{itemize}

We consider the following three use cases.

\begin{enumerate}

\item \label{uc:transactions} Communication of a pending transaction to the committee $C_{i,r}$ that is
in charge to process it in round $r$, where $i$ is obtained by the transaction
itself with a function that it is not important in this context (e.g., a hash
function).

\item  Communication within each committee to perform a consensus algorithm.

\item \label{uc:pipeline} Communication of an intermediate result from a committee $C_{i,r}$ to a
committee $C_{j,r'}$, where $r<r'$, to be used in pipeline-based architectures
like in~\cite{DBLP:journals/corr/abs-2005-06665}.

\end{enumerate}

Since the scalability objective implies that transactions should not be
broadcasted, we have the inherent constraint that the destination committee
should be known when message is sent. Hence, for use case~\ref{uc:pipeline},
$C_{j,r'}$ should be reachable $r'-r$ rounds earlier with respect to when
$C_{j,r'}$ actually will perform consensus. Also in use
case~\ref{uc:transactions}, there is a similar constraint, in the sense that
transactions to be processed by a committee should be sent in advance. This
observation is tightly linked with our threat model.

\subsection{Threat Model}\label{threat-model}

Any consensus protocol that tolerates Byzantine faults is reliable up to a
certain fraction of Byzantine members. In other words, to subvert a consensus
instance run by a committee $C$ the attacker should control at least a certain fraction of $C$. 
In a \textit{Sybil attack}, the attacker creates a huge number of
fictitious \emph{sybil-nodes}, making high its probability to control any
randomly chosen committee. We assume that an opponent can control only a limited
number $m$ of sybil-nodes, with $m$ independent of the number $N$ of potential
committee members. The logic behind this assumption is the possibility to defend
against Sybil attacks imposing some sort of cost for each
sybil-node (e.g., like in proof-of-stake or proof-of-work~\cite{urdaneta2011survey} approaches) and assuming that the attacker can bear
only a limited cost independent of $N$.

We assume a slowly-adaptive adversary, i.e., the period of time required for the
adversary to change the set $m$ of controlled nodes to match a selected
committee is longer than the period of time between the instant when a committee
is known and the instant when the committee reaches consensus. Clearly, this model
is as much sensible as the frequency at which committees change is high and
$r'-r$ is low.

If an attacker were to succeed in taking control of a committee, subversion of
the regular consensus rules are possible. Further, in our communication related
setting, the threats are the followings:
\begin{itemize}
    \item \textbf{Message forging}: the attacker can decide the message to be sent or pretend that a certain message was received.
    \item \textbf{Denial of service (DoS)}: messages are not received or sent.
\end{itemize}

\subsection{Committee Selection}\label{sec:requirements:selection}

The committee selection procedure must satisfy the following requirements:
\begin{itemize}

    \item \textbf{Unpredictability.} Members of selected committees should be
    hard to predict before the instant $r$ in which the
    committee is selected (see use cases~\ref{uc:transactions} and~\ref{uc:pipeline} described above). Further, it should be hard for the attacker to steer
    the selection toward the nodes it already controls. 
    
    \item \textbf{Uniformity.} The probability for a subject to be selected as a
    committee member should be uniform across all potential committee members.
    Otherwise, an attacker could try to control the subjects that have more
    probability to be part of a committee.

    \item \textbf{Verifiability.} Given a message, anyone should be able to determine if the sender of the message is a legitimate member of a certain committee.
 
    \item \textbf{Robustness.} The desired committee size $k$ should be a
    parameter given to the selection procedure. Let the actual size be $k'$.
    Ideally, we want $k'=k$,  however, in practice, we can accept $k' \in
    [k,\alpha k]$ (with $\alpha>1$). We can also have a randomized approach, provided that the
    probability of $k'<k$ is negligible. Note that Having $k'>k$ does not affect
    security, but a large $k'$ may affect efficiency. Tuning $\alpha$ can help
    us to strike a good trade-off.

\end{itemize}

The robustness requirement is inspired by a similar requirement adopted
in~\cite{pennino2021efficient}.

\subsection{Communication}

Broadcasting all messages greatly simplifies communication but severely limits
scalability. On the other hand, in traditional unicast messaging (1) the sender
have to know the destination address and (2) a routing mechanism (and
corresponding routing tables) should be available. Since, we aim at changing
committee frequently and at involving most of the nodes in at least one
committee, explicitly publishing all committees to all nodes is not feasible.
Further, adapting routing when committee changes is not feasible either.

We aim at a communication protocol with the following requirements.

\begin{itemize}

    \item \textbf{Routing-by-committee-ID.} The destination committee of a
    message should be expressed as a small committee identifier. Committee identifiers
    should be pairs $\langle i,r \rangle$, where $r$ is the \emph{round index}
    and $i \in [0,K)$ integer, identifies each committee in round $r$.  No
    explicit publishing of committee members should be performed.
    
    \item \textbf{No-rerouting.} No routing information propagation should be needed when committees change.
    
    \item \textbf{Log-size.} $O(\log N)$ routing table size.
     
    \item \textbf{Log-time.} $O(\log N)$ time to route a message to a committee. 
\end{itemize}

Note that property routing-by-committee-ID implies that partitioning the load
related to processing pending transactions is easy. In fact, suppose a
transaction $x$ should be assigned to a committee $i$ in round $r$, this can be
done by defining an arbitrary (non-cryptographic) hash function $f_h$ from $x$
to integers in $[0,K)$, then $x$ can be sent to the committee identified by
$\langle f_h(x),r \rangle$ (see also Section~\ref{communication}).

\section{Solution} \label{sec:solution} 

In this section, we describe a committee
selection procedure and a methodology for inter-committee communication that
meet the requirements described in Section \ref{sec:requirements}. Our approach
leverages the kademlia peer to peer network described
in~\cite{maymounkov2002kademlia} and adapts it to our needs.

In our approach, each node are part of a kademlia-like peer-to-peer network.
While kademlia was designed for data retrieval, we use it for a specific kind of
multicast routing. In kademlia, each node has an identifier called \emph{kad-ID}, or
\emph{KID}, which can be considered random numbers (for honest nodes) from 0 to $2^b-1$
(where $b$ is the number of bits of their binary representation). We call
\emph{KID space} the set of possible KIDs. Often $b=160$, however, theoretically, $b$ is more
correctly chosen to be $O(\log N)$ so that the average density of
used KIDs in the space is constant when $N$ grows. Usually, in a blockchain,
each node is associated with a public/private key pair and we assume its KID is
a cryptographic hash function of its public key.

In kademlia, the space of the KIDs is a metric space. The distance $d(n_1, n_2)$
between two KIDs $n_1$ and $n_2$ is defined as the XOR operation between the
binary representation of $n_1$ and $n_2$ interpreted as an integer number. We
refer to this as the \emph{xor metric}. For simplicity, in the following, we
refer with the same symbol to both a node and its KID.

Kademlia is able to find the IP address of nodes close to a certain KID value in
time $O(b)$, i.e., $O(\log N)$, with routing table size for each node also
$O(\log N)$. Details can be found in~\cite{maymounkov2002kademlia}.

The xor metric has the following property.

\begin{property}\label{pr:unique-at-distance}
Let $d(\cdot, \cdot)$ be the xor metric, $v$ a KID, and $q$ an integer,
there exist only one KID $u$ such that $d(u,v)=q$.
\end{property}

Which implies the following.

\begin{property}\label{pr:count-at-distance}
Let $d(\cdot, \cdot)$ be the xor metric, $v$ a KID, and $q$ an integer,
there number of KIDs $u$ such that $d(u,v)<q$ are $q$.
\end{property}

\subsection{Committee Selection Procedure} \label{selection-procedure}

We now describe the procedure to randomly select a committee with identifier
$\langle i, r \rangle$. 

We assume $R_r$ to be an unbiased random number available during round $r$. If
the blockchain is one, as in~\cite{DBLP:journals/corr/abs-2005-06665}, and the
network is so big that transactions in the block are hard to control by any
attacker, $R_r$ can simply be the hash of the block produced in round
$r-1$. For simplicity we assume instantaneous diffusion of this hash. For
sharding-based blockchains, the sequence of $R_r$ can be produced by pseudo
random generation periodically seeded with an unbiased random number produced in
a decentralized manner by a \emph{reference committee}. For example,
RapidChain~\cite{zamani2018rapidchain} adopts \emph{verifiable secret
sharing}~\cite{feldman1987practical} to do so.

We now describe the selection procedure. Suppose that, for a certain blockchain
technology, committees that will be active at round $r$ have to be contacted no
more than $e$ rounds earlier (i.e., at round $r-e$). Members of a committee
$C_{i,r}$ are selected in the following way. Consider $\gamma$ pseudo-random
values $v^j_{i,r} = \mathrm{hash}( R_{r-e} , i, j)$ with $\gamma$ integer, $j$
in $1,\dots,\gamma$, $v^j_{i,r}$ in the KID space, and $\mathrm{hash}(\cdot)$ a
cryptographic hash function. For simplicity, in the following, we omit the
subscript, hence, we define $v^j=v^j_{i,r}$ and $C=C_{i,r}$.

We create $C$ as the
union of $\gamma$ sets of nodes, i.e., $C = \bigcup_j C^j $, each
depending on the corresponding $v^j$.

We denote as $ \beta = \frac {k(1+\alpha)} {2\gamma} $  the desired expected
value of the number of nodes in $C^j$, where $k$ is the minimum number of nodes
for a committee and $\alpha>1$ is a parameter that can be tuned to make the
selection robust. The rationale for this formula is that $(k+\alpha k)/2$ is the
middle point of range $[k,\alpha k]$ of the desired size for a committee to
achieve the robustness property described in
Section~\ref{sec:requirements:selection}. The formula considers that the
contribution to each committee is provided by $\gamma$ spots. Robustness is
proven in Section~\ref{sec:robustness-analysis}.

We define $C^j = \{ n \in U : d(n,v^j)<\frac{2^b}{N} \beta\}$. We call
\emph{spot} associated with $v^j$ both $C^j$ and the set of KIDs within distance
$\frac{2^b}{N} \beta$ from $v^j$. The actual meaning will be clear from the
context. We call $v^j$ the \emph{center} of $C^j$ and we say that $C^j$ is the spot
of $v^j$.

Further, for security reasons, we impose a limit $l= k/\gamma^\rho$ (with $0<\rho<1$) so that
$|C^j| \leq l $ , in the sense that only the $\floor*{ l }$
nodes closer to $v^j$ are taken to be in $C$. The idea is that an attacker
cannot accumulate all its bad nodes in a certain zone in the hope that a single
spot overlaps all of them. On the contrary, the more spots we have the more the 
attacker have to control a larger number of them at the same time. 
This reduces the probability of a successful attack.

\subsection{Communication}\label{communication}

An honest node $u$ that intends to send a message to all members of committee
$C_{i,r}=\bigcup^\gamma_j C_{i,r}^j$ can autonomously compute values
$v^j_{i,r}$. To simplify the notation, we omit subscripts $i,r$. By construction,
members of each spot $C^j$ are those within a certain distance from  $v^j$ and
if they are more than $\floor*{l}$, only the closest ones are in $C^j$. 

In our approach, we adopt the well known kademlia iterative search to deliver
the message to one or more members of each spot~\cite{maymounkov2002kademlia},
then gossiping is used to diffuse the message within the spot.

In detail, to deliver a message to all memebers of $C$, $u$ starts $\gamma$
kademlia searches where the $j$-th search is targeted to $v^j$ (with
$j=1,\dots,\gamma$). The kademlia routing system is designed to reach nodes that
are close to a certain KID value. Consider any kademlia node $n$, we also denote
$n$ its KID. One notable property of the routing table kept by $n$ is that it
stores the KIDs of a set of nodes whose density is exponentially decreasing as
the distance from $n$ increases. This means that $n$ knows the nodes close to it
with great precision. When a node $n' \in C^j$ is reached by the $j$-th search,
$u$ can deliver the message to it by the underlying network. By the properties
of the kademlia search and  kademlia routing tables, and by the construction of
each spot $C^j$, it is very likely that the message can be gossiped from $n'$ to
all other honest members of $C^j$. In fact, all members $C^j$ knows $v^j$ and
can easily select the nodes in their kademlia routing table that are part of
$C^j$ (by verifiability requirement, see Section~\ref{selection-procedure}) and
can use this information to execute the gossip protocol within $C^j$.

Requirement Routing-by-committee-ID is fulfilled by construction. Requirement
no-rerouting is fulfilled since routing information depends only on node KIDs
and not on committees selection. Requirements Log-size and
Log-time are fulfilled since they are properties of the kademlia routing scheme.

The described approach works well for one-shot communications (e.g., for
submitting a transaction to a committee). However, during consensus, members of
the same committee $C= \bigcup^\gamma_j C^j$ should communicate several times
and performing kademlia search for every inter-spot communication is too
cumbersome. To speed up communications during consensus, all members of each
spot $C^j$ send to all other spots the list of the members of $C^j$. This is
done before consensus algorithm starts. When this information is received, it is
gossiped within the spot as above. After this preparation phase, the consensus
can be executed at the speed of the underlying network connection. Note that,
this preparation can be executed in parallel with the reception and processing
of the data needed to be processed before the execution of the consensus
algorithm. Hence, it affects network load but not latency.

Clearly the efficiency of our communication protocol badly depends on the number 
of spots adopted for each committee. In Section~\ref{sec:experiments}, we show that 
the probability of a successful attack quickly decrease even with small values of $\gamma$.

\section{Robustness Analysis} \label{sec:robustness-analysis}

We formally state the robustness of our approach in the following theorem. 

\begin{theorem}[Robustness]
    Consider the selection of a committee $C$ form a universe $U$ of potential
    members. The selection procedure described in Section~\ref{selection-procedure}
    with desired committees size $k$ and a parameter $\alpha>1$, is robust
    in the sense that it is possible to arbitrarily reduce the probability of having
    $|C|<k$ by increasing $\alpha$, if the number of potential committee members $N=|U|$ is large. 
\end{theorem}
    
    This result is somewhat similar to what reported in Theorem 1
    of~\cite{pennino2021efficient}. It can be proved as follows.
\begin{proof}    
    Each honest member has a uniformly distributed KID which is independent of
    the random center $v^j$ of $C^j$.  Note that, $|C|$ is a sum of $\gamma$
    independent identically distributed random variables, one for each $C^j$.

    For each spot, the number of KIDs $n$ such that
    $d(n,v^j)<\frac{2^b}{N}\beta$ is  $ \frac{2^b}{N}\beta$, by
    Property~\ref{pr:count-at-distance}. The probability that a node is in the
    spot of $C^j$ is $p=\frac{\frac{2^b}{N}\beta}{2^b}=\frac{\beta}{N}$ and the
    number of nodes in the spot (without considering the limit $l$) is a random
    variable $Y$ with binomial distribution with expected value $\mu_Y=Np=\beta$
    and variance $\sigma^2_Y=N\frac{\beta}{N}(1-\frac{\beta}{N})$. Hence, the
    distribution of the random variable $|C^j|$ can be approximated with a
    ``trucated'' binomial distribution where the weight of the upper tail above
    $\floor*{l}$ is accumulated on the $\floor*{l}$ value.

	Following the theorem statement, we consider an increasing $\alpha$ parameter.
	When $\alpha$ increases also $\beta = \frac {k(1+\alpha)} {2\gamma} $
	increases. For large $N$, we can approximate the binomial distribution with
	normal distribution with same mean and variance. If we denote with
	$\Phi(\cdot)$ the standard normal cumulative distribution, $P[Y<l]\simeq
	\Phi\left(\frac{l-\beta}{\sqrt{\beta(1-\beta/N)}}\right)$. The argument of
	$\Phi(\cdot)$ decreases when $\beta$ increases, for $\beta \ll N$ , hence, by
	monotonicity of $\Phi(\cdot)$, $P[Y<l]$ decreases and tends to zero. This means
	that, if $\alpha$ is big enough, $|C|=\sum_{j=1}^{\gamma}
	|C^j| \simeq \sum_{j=1}^{\gamma} l = \gamma l = \gamma \floor*{k/\gamma^\rho}\simeq
	\gamma^{1-\rho} k$, which proves the theorem.
	\end{proof}

\section{Security analysis}\label{sec:security-analisys}

As explained Section~\ref{threat-model}, to perform any attack, the attacker must control a committee. We call \emph{bad committee}, a
committee where the attacker controls a number of members above a certain
threshold $\bar{k}$. We call \emph{bad nodes} the nodes controlled by the attacker.

We adopt the following notation.
\begin{itemize}
    \item $C$ is a pseudo-random committee selected with the approach described in Section~\ref{selection-procedure}.
    \item $U$ is the set of \emph{potential} members for $C$.
    \item $N=|U|$ is the cardinality of $U$. 
    \item $A$ is the set of the bad nodes in $U$.
    \item $\tilde{A} = C \cap A$ is the set of bad nodes that occurred to be in $C$.
    \item $\tilde{k}=|\tilde{A}|$ is the cardinality of $\tilde{A}$.
    \item $\bar{k}$ is the number of honest members that should agree to achieve consensus according to a certain consensus algorithm, with $\bar{k}< k = |C|$.
\end{itemize}

We formalize the security of our approach in the following theorem.

\begin{theorem}[Security]

For the selection committee procedure, described in Section~\ref{selection-procedure},  the
probability to obtain a bad committee tends to zero for increasing $N$ 
under the following conditions: 
\begin{enumerate}
\item the attacker that can control a slowly varying set of nodes of constant size,
\item the number of spots per committee $\gamma$ is $\Omega(\log N)$,
\item the number of committees is $O(N)$
\end{enumerate}

\end{theorem}

\begin{proof} 
Let $m$ be the number of potential committee members that an attacker can
control, i.e., $m=|A|$. Any
committee $C=\bigcup_j C^j$ is fixed and reachable from any node when random seed $r$ is known. The attacker choose $A$ by choosing the KID of each bad
node. The attacker is free of placing them anywhere in the KID space. The slowly varying hypothesis,
allows the attacker to make any choice of $A$, but not to adapt $A$ once $r$ is
known. The attack is successful when a bad committee is selected and this occurs
when $ \tilde{k} \geq \bar{k}$, where $\tilde{k}= \sum_j \tilde{k}^j$ and $\tilde{k}^j$ is the number of bad nodes in spot $C^j$ that are taken as part of $C$ (remember that $|C^j| \leq l$ by construction, see Section~\ref{selection-procedure}). 

The minimum effort for the attacker to obtain a bad committee is to arrange
things such that $\tilde{k}=\bar{k}$ bad nodes are selected in $C$, since no
success is obtained with fewer bad nodes in $C$. Since the selection procedure
of $C$ impose $\tilde{k}^j \leq l$, we consider the following attacker strategy.
The attacker creates \emph{clusters} of $\floor*{l}$ bad nodes that are very
close together (according to the xor metric) and spread these clusters randomly,
but evenly, in the KID space. Intuitively this is an optimal strategy. In fact,
for reasonable parameters ($\bar{k}< m \ll N$), if one of the nodes of a cluster
$\chi$ is in $C^j$ the probability that the whole $\chi$ is in $C^j$ is close to
1. Spreading the bad nodes in $\chi$ decreases this probability and hence the
probability of a successful attack. Increasing the size of $\chi$ does not
provide any advantage, but reduces the nodes available for other clusters.
Decreasing the size of $\chi$ makes the attack harder since more clusters are
needed to reach $ \tilde{k} \geq \bar{k}$.

According to this strategy, with $m$ bad nodes, an attacker can create $z=
\floor*{ \frac{m}{\floor*{l}}}$ clusters of size $\floor*{l}$.

The attack is successful if there are $\bar{z}=\ceil* {\frac{\bar{k}}{\floor*{l}}} \simeq \gamma^\rho \bar{k}/k$ clusters $\chi$ such that, for each $n \in \chi$,
$d(n,v^j)<\frac{2^b}{N} \beta$. By Property~\ref{pr:count-at-distance}, there
are $ \floor*{\frac{2^b}{N} \beta }- \floor*{l}$ values of $v^j$ for which
$\chi$ is entirely contained in spot $C^j$.
Hence, the probability that a specific spot contains one of the clusters is $p=
z \frac{\floor*{2^b\beta/N} - \floor*{l}}{2^b} $ $\simeq \frac{z \beta}{N} $
$= \frac{ \floor*{\frac{m}{\floor{l}}} \beta}{N}$
$\simeq \frac{m(1+\alpha)}{ 2N\gamma^{1-\rho} }$.

Since spots are independently selected. The number $X$ of occurrences of this
event across all spots $C^j$, with $j=1\dots\gamma$, is a binomially distributed
random variable with mean $\mu=\gamma p$. 

To prove the theorem statement, note that the attacker is attacking the creation of 
$O(N)$ committees at the same time. Since each committee is created independently, 
the total probability of a successful attack is $\Pi \sim N \cdot P[X > \bar{z}]$.

The second factor can be bound by applying the Chernoff bound $P[X> (1+\delta)
\mu] <  \left( \frac{e^\delta}{(1+\delta)^{1+\delta}}  \right)^\mu$ with
$(1+\delta) \mu = \bar{z}$. 
Performing substitutions, and posing $q=\frac{2\bar{k}}{m k(1+\alpha)}$, we obtain
$N\cdot P[X> \bar{z}]< N \left( \frac{e^{qN-1}}{(qN)^{qN}} 
\right)^{\frac{m(1+\alpha) \gamma^\rho}{2N}} $. Since for the theorem statement
$\gamma = \log N $, we substitute $N=e^\gamma$, obtaining $e^\gamma \left(
\frac{e^{qe^\gamma-1}}{(qe^\gamma)^{qe^\gamma}}  \right)^{\frac{m(1+\alpha)
\gamma^\rho}{2e^\gamma}} $ $\sim e^\gamma \left(
\frac{e^{qe^\gamma}}{(qe^\gamma)^{qe^\gamma}}  \right)^{\frac{m(1+\alpha)
\gamma^\rho}{2e^\gamma}}$ $= e^\gamma \left(
\frac{e}{qe^\gamma}  \right)^{s\gamma^\rho}$, where $s=\frac{m(1+\alpha) q}{2}$.
Rearranging, we obtain
$=\frac{e^{s\gamma^\rho+\gamma - s\gamma^{\rho+1}}}{q^{s\gamma^\rho}} $
$\sim \frac{e^{- s\gamma^{\rho+1}}}{q^{s\gamma^\rho}}$ $= \frac{1}{(q e^\gamma)^{s\gamma^\rho}} \rightarrow 0$, for $\gamma \rightarrow \infty$ (and hence $N\rightarrow\infty$), which proves
the statement.
\end{proof}

In Section~\ref{sec:experiments}, we experimentally show that the probability of having a round with 
bad committee quickly tends to zero even for a small number of spots.

\section{Experiments}\label{sec:experiments}

We performed some initial experiments to assess how the probability of a successful attack behaves 
when varying the number of nodes $N$ and the number of  spots per committee $\gamma$. 

We simulated committee selections in a universe of nodes where the adversary can
control a number of nodes much larger than the typical size of the committee. 
We show the results
of our experiments in Table~\ref{tbl:results}. We performed our experiments with
the following parameters. Number of nodes $N$ equal to $1000$, $3000$ and $10000$.
Number of committees $K=N/(2k)$ of desired size $k=20$ (i.e., affordable for BFT
algorithms), which means about half of the nodes are involved in committees at each round. Number of malicious nodes $m=200$. Number of spots per committee $\gamma$ from 1 to 5.
Additional parameters are $\alpha=3$, $\bar{k}=2k/3 $, $\rho=0.9$. Probability
estimation is performed by executing $N\gamma/10$ tests for each cell of the table. Each cell shows the probability of having a round of $K$ committees with at least
one committee under the control of the attacker.

We have found that in practice the probability of a successful attack quickly
decreases even with a small number of spots ($\gamma$), which makes feasible the adoption of the 
communication protocol shown in Section~\ref{sec:solution} .
We note that for $\gamma=1$ probability of a successful attack is quite high. 
Further, probability quickly decreases by increasing $N$ for $\gamma \geq 2$.
Note that, the reported numbers are probabilities estimated by repeatedly performing committee selection and checking for a successful attack (with the parameter given above). For this reason, the two zeros in the rightmost column should be interpreted as very low probabilities, with respect to the precision of the estimation.

\begin{table}
\centering
\begin{tabular}{|c|*{5}{c|}}
      \hline
      \diagbox[width=1cm, height=0.6cm]{ $N$ }{$\gamma$ }
                   & $1$ & $2$ & $3$ & $4$ & $5$ \\
      \hline
    $1000$ &$0.8$ & $0.25$& $0.08$ &  $0.035$& $0.006$  \\
    \hline
      $3000$  & $0.81$ & $0.081$ & $0.03$ & $0.0016$  & $0.0$ \\
    \hline
          $10000$ & $0.853$ & $0.026 $ &  $0.0053$ & $0.001$  & $0.0$ \\
      \hline

    \end{tabular}
\vspace{0.2cm}

\caption{Experiments results. The probability of a successful attack tends quickly to zero even 
for small values of $\gamma$.}
\label{tbl:results}
\end{table}

\section{Conclusions and Future Work}\label{sec:conclusions}

In this paper, we proposed a committee selection procedure and an associated
communication procedure, that enable a scalable and efficient committee-targeted
communication to be used in a blockchain architecture. In particular, the selection procedure has no negative effects on
decentralization, in fact the members of each committee are selected in a statistically
unbiased way. Further, the whole proposal is designed so that committees can be frequently changed, so that all the nodes can
statistically cooperate in the creation of a fraction of new blocks.

Regarding the communication procedure, we proposed a routing based on the
kademlia algorithm, in which messages are targeted to the centers of
agglomerates of nodes, called spots, that are in the committee by construction,
reach one node of the spot in $O(\log N)$ time, where $N$ is the number of
nodes, and then are gossiped within the spot (in constant time since the spot
has constant size).

We formally proved  the security of the committee selection procedure against a
slowly adaptive adversary and we experimentally proved that the probability of a successful attack
quickly tends to zero even for a handful of spots.

Our approach enable blockchain scalability even in the presence of rapidly
changing committees and is compatible with models in which nodes have limited
resources and scalability should be achieved exclusively by increasing the
amount of nodes. In fact, messages need not to be transmitted in broadcast but
are multicasted to fresh new committees without the need to change any routing
table in the nodes. This makes our approach ideal for sending pending
transactions to the right committee, for sharding-based architectures, like
in~\cite{zamani2018rapidchain}, or to send partial result to another committee,
for pipeline-based architectures, like in~\cite{pennino2021efficient}.

As a future work, we plan to implement and perform realistic experiments with the proposed approach to verify and measure its efficiency in practice. We plan to do that in the context of experimenting the pipeline-based architecture described in~\cite{pennino2021efficient}.

\newcommand{\etalchar}[1]{$^{#1}$}

\end{document}